\documentclass[11pt]{article}

\usepackage[margin=1in]{geometry}
\usepackage{amsmath, amssymb, amsthm, mathtools}
\usepackage{bm}
\usepackage{hyperref}
\usepackage{enumitem}
\usepackage{algorithm}
\usepackage{algpseudocode}
\usepackage{microtype}
\usepackage{mathpazo}
\usepackage{todonotes}
\usepackage{xspace}

\usepackage{subcaption}
\usepackage{thmtools}
\usepackage{thm-restate}
\usepackage[dvipsnames]{xcolor}
\usepackage{tabularray}
\usepackage{siunitx}
\usepackage[british]{babel}

\newtheorem{theorem}{Theorem}
\newtheorem{lemma}{Lemma}
\newtheorem{proposition}{Proposition}
\newtheorem{definition}{Definition}
\newtheorem{remark}{Remark}

\newcommand{\C}{\mathbb{C}}

\newcommand{\ket}[1]{\left|#1\right\rangle}
\newcommand{\bra}[1]{\left\langle#1\right|}
\newcommand{\braket}[2]{\left\langle#1\,\middle|\,#2\right\rangle}
\newcommand{\norm}[1]{\left\|#1\right\|}
\newcommand{\abs}[1]{\left|#1\right|}

\newcommand{\Ibox}[1]{\left[#1\right]}

\newcommand\qksat{$k$-QSAT\xspace}
\newcommand\twoqsat{$2$-QSAT\xspace}

\usepackage{cancel}
\usepackage{tikz}
\usetikzlibrary{arrows.meta}
\renewcommand\phi\varphi

\algnewcommand\Input{\Statex \textbf{Input:}\ }

\title{Search-Driven Clause Learning\\ for Product-State Quantum $k$-SAT (PRODSAT-QSAT)}
\author{Samuel González-Castillo\and Joon Hyung Lee \and Alfons Laarman }
\date{\today}

\begin{document}
\maketitle

\begin{abstract}
We study PRODSAT-QSAT($k$): given rank-one $k$-local projectors, determine whether a quantum $k$-SAT instance admits a satisfying product state. We present a CDCL-style refutation framework that searches a finite partition of each qubit’s Bloch sphere while a sound theory solver checks region feasibility using a geometric over-approximation of the projection amplitudes for each constraint. When the theory solver proves that no state in a region can satisfy a constraint, it produces a sound conflict clause that blocks that region; accumulated blocking clauses can yield a global result of product-state unsatisfiability (UN-PRODSAT). We formalise the problem, prove the soundness of the clause-learning rule, and describe a practical algorithm and implementation.
\end{abstract}

\tableofcontents

\section{Introduction}
Quantum $k$-SAT ($k$-QSAT) asks whether there exists an $n$-qubit state annihilated by a set of local projectors $\Pi_j$. Equivalently, given a Hamiltonian $H \coloneq \sum_j \Pi_j$, the problem asks whether $H$ is frustration-free, i.e., if its ground energy is zero.
The subject is a quantum analog of the classical problem of $k$-SAT, but much less explored.  There have been successful discoveries.  It was first studied by Kitaev~\cite{kitaev}.  \twoqsat has been proven to be easy to solve in \cite{Bravyi,beaughar}.  The condition to have product state solutions and the algorithm to have such solutions have been proven in \cite{prodsat2024}, which still leaves many questions to explore.  For instance, the algorithm given in \cite{prodsat2024} is doubly exponential.  It also opens questions about the phase beyond product state phase, namely entangled state and UNSAT. In analogy to how $k$-SAT is NP-complete for $k \geq 3$, $k$-QSAT is QMA-complete for $k\geq 3$, as shown in~\cite{gosset2016quantum}.

Even for small sizes, heuristic solvers must navigate a nonconvex, high-dimensional landscape.
From an algorithmic perspective, it is natural to seek satisfiable states within structured ans\"atze.
In this work we focus on the \emph{product-state} ansatz, which is both physically meaningful and algorithmically tractable.  Specifically, we aim to certify the absence of product state solutions (UN-PRODSAT) of a finite size \qksat leveraging the principle of conflict-driven clause learning (CDCL)~\cite{marques2002grasp,marques2009conflict}.

The product state is defined as follows:
\begin{equation}\label{def:prodsat}
\ket{\psi} = \bigotimes_{j=1}^n \ket{\psi_j}, \qquad \ket{\psi_j}\in\C^2,\ \norm{\ket{\psi_j}}=1.
\end{equation}
We develop a clause-learning framework for this setting, that helps guide a dichotomic search over the state space.
The key idea is to discretize the single-qubit manifold into finitely many regions
and let a CDCL SAT solver~\cite{marques2002grasp,marques2009conflict,cdcl} search over regions while a continuous theory solver checks feasibility
inside a region using interval arithmetic. When feasibility fails, a \emph{sound conflict clause} can be produced.

\paragraph{Problem definition.}

Fix $n$ qubits. For each constraint $j\in[m]$, let $S_j\subseteq [n]$ be a subset with $\abs{S_j}=k$ and let $\Pi_j$ be a rank-one projector supported on $S_j$:
\[
\Pi_j \;=\; \ket{v_j}\bra{v_j}, \qquad \ket{v_j}\in (\C^2)^{\otimes k},\ \norm{\ket{v_j}}=1.
\]
For a global state $\ket{\psi}\in(\C^2)^{\otimes n}$, let $\ket{\psi_{S_j}}$ denote the projected state on $S_j$. We interpret $\Pi_j$ as excluding the local state $\ket{v_j}$: a global state $\ket{\psi}$
satisfies constraint $j$ if and only if
\[
\Pi_j \ket{\psi} = 0 \quad \Longleftrightarrow\quad \braket{v_j}{\psi_{S_j}}=0.
\]

\begin{definition}[PRODSAT-QSAT($k$)]
Given rank-one $k$-local projectors $\Pi_1,\dots,\Pi_m$, decide whether there exists a product state
$\ket{\psi}=\bigotimes_{j=1}^n\ket{\psi_j}$ such that $\Pi_j\ket{\psi}=0$ for all $j\in[m]$.
\label{def:pqsat}
\end{definition}

\paragraph{Contributions.}
\begin{itemize}[leftmargin=2em]
\item We formalize \emph{PRODSAT-QSAT($k$)} and an algorithm architecture tailored to certify product-state unsatisfiability.
\item We give a sound clause-learning rule based on interval exclusion of zero from projection amplitudes.
\item We propose practical rules for clause learning and provide an implementation of our algorithm.
\end{itemize}

In particular, we introduce a theory solver (Algorithm~\ref{alg:constraint}) that can detect product-state unsatisfiability of individual QSAT($k$) constraints in a family of regions in the product state space. Building on this, we use our clause-learning framework to define an algorithm (Algorithm~\ref{alg:full}) that is backed by the following result:

\begin{theorem}
Fix $k$ and consider a QSAT($k$) instance. Algorithm~\ref{alg:full} runs, in the worst-case scenario, in exponential time in the number of qubits. It returns UN-PRODSAT if it can certify that its input instance is UN-PRODSAT. Otherwise, it returns MAYBE($A$, $\rho$), where $A$ and $\rho$ are the sums of, respectively, the area and the maximum modulus squared of the set over-approximations constructed by the theory solver.
\label{thm:main}
\end{theorem}

In the result above, the parameters $A$ and $\rho$ act as heuristic indicators of how likely the instance is to be PRODSAT when the algorithm cannot certify that it is UN-PRODSAT. Smaller values of those parameters would imply a higher likelihood.

\paragraph{Paper outline.}
\begin{itemize}
\item In Section~\ref{sec:geometry}, we develop the geometric tools used by the theory solver.
\item In Section~\ref{sec:tsolver}, we introduce the theory solver (Algorithm~\ref{alg:constraint}) and discuss and prove its soundness guarantees (summarised in Proposition~\ref{prop:theory-solver}).
\item In Section~\ref{sec:main}, we discuss the overall CDCL-based architecture and the clause-learning rule, and we prove Theorem~\ref{thm:main}.
\item In Section~\ref{sec:implement}, we discuss our implementation and share empirical observations.
\item Finally, we summarise our results in Section~\ref{sec:concl}.
\end{itemize}

\section{Geometric preliminaries}
\label{sec:geometry}

With a product state $\ket{\psi}$, write each single-qubit state (up to global phase) in Bloch coordinates:
\[\ket{\psi_j} = \ket{\Psi(\theta_j,\varphi_j)} \;=\; \cos(\theta_j/2)\ket{0} + e^{i\varphi_j}\sin(\theta_j/2)\ket{1},
\qquad \theta_j\in[0,\pi],\ \varphi_j\in[0,2\pi).\]
For a constraint $\Pi_i = \ket{v_i}\bra{v_i}$ supported on $S_i=\{i_1,\dots,i_k\}$,
$\Pi_i\ket{\psi}=0$ is equivalent to
\[\bra{v_i}\cdot \bigotimes_{j\in S_i} \ket{\Psi(\theta_j,\varphi_j)}.\]

\newcommand\conv{\operatorname{conv}}
\newcommand\poly[1]{\operatorname{poly}[#1]}
\newcommand\ann{\operatorname{Ann}}
\newcommand\polyenc[1]{\operatorname{polyenc}\left(#1\right)}

Given some complex numbers $p_1, \ldots, p_\ell$ (with $\ell\geq 2$), we write $\poly{p_1,\ldots,p_\ell}$ to denote the polygon in the complex plane with edges $[p_1, p_2], \ldots, [p_{\ell-1}, p_\ell], [p_\ell, p_1]$. Notice how we consider segments to be polygons (when $\ell = 2$). Moreover, given any real numbers $r_0,r_1,\phi_0,\phi_1$ with $0 \leq r_0 \leq r_1$ and $\phi_0 \leq \phi_1$, we will consider the annular sector
\[ \ann(r_0,r_1;\phi_0,\phi_1) \coloneq [r_0, r_1]e^{i[\phi_0, \phi_1]},\]
which is depicted in Figure~\ref{fig:annular}.

Given any two sets $A, B \subseteq \mathbb{C}$, we define their Minkowski sum and product to be, respectively,
\[ A \oplus B \coloneq \{ a + b \mid a\in A, b \in B\}, \qquad A \odot B \coloneq \{ a \cdot b \mid a \in A, b \in B \}.\]
Moreover, for any set $A\subseteq \mathbb{C}$, we write $\conv A$ to denote its \emph{convex hull}: the smallest convex subset of $\mathbb{C}$ containing $A$. Of course, in particular, for any points $p_1, \ldots, p_n \in \mathbb{C}$, we have~\cite{convex-sets}
\[\conv\{p_1, \ldots, p_\ell\} \coloneq \left\{ \sum_{r = 1}^\ell \lambda_r p_r \, \middle\vert\, \lambda_r \in \mathbb{R}_{\geq 0}, \sum_{r = 1}^\ell \lambda_r = 1\right\},\]
and, if $\poly{p_1, \ldots, p_\ell}$ is a convex polygon, then $\poly{p_1,\ldots,p_\ell} = \conv\{p_1, \ldots, p_\ell\}$.

If a state $\ket{\psi}$ of the form $\ket{\psi(\theta, \phi)}$ satisfies $\theta \in [\theta_{0}, \theta_{1}]$ and $\phi \in [\phi_{0}, \phi_{1}]$, then
\[ \braket{0}{\psi} \in \cos([\theta_{0}/2, \theta_{1}/2]), \qquad \braket{1}{\psi} \in \sin([\theta_{0}/2, \theta_{1}/2]) e^{i[\phi_{0}, \phi_{1}]},\]
hence the amplitude $\braket{0}{\psi}$ lies in the interval $[s_{0}, s_{1}] \coloneq  \cos([\theta_{0}/2, \theta_{1}/2])$, which can be identified as the annular sector $\ann(s_0,s_1;0,0)$, whereas the amplitude $\braket{1}{\psi}$ lies in the annular sector $\ann(r_0,r_1;\phi_0,\phi_1)$ with $[r_0, r_1] \coloneq \sin([\theta_0/2, \theta_1/2])$.
Therefore, both amplitudes take values in annular sectors.

Computing the Minkowski product of annular sectors is very straightforward, as the following result shows.

\begin{lemma}
Let $\ann(r_0,r_1,\phi_0,\phi_1)$ and $\ann(s_0,s_1,\theta_0,\theta_1)$ be annular sectors. Their Minkowski product can be computed as
\[ \ann(r_0,r_1,\phi_0,\phi_1) \odot \ann(s_0,s_1,\theta_0,\theta_1) = \ann(r_0 s_0, r_1 s_1, \phi_0 + \theta_0, \phi_1 + \theta_1).\]
\label{lem:product-hull}
\end{lemma}
\begin{proof}
Taking into account that addition and multiplication are continuous and monotone, we only need to consider the fact that closed intervals are connected and compact, and that, for any real numbers $r,s,\theta,\phi$, we have $r\exp(i\phi) \cdot s \exp(i\theta) = rs\exp(i(\phi + \theta))$.
\end{proof}

\begin{figure}[htbp]
  \centering
  \begin{subfigure}[t]{0.45\textwidth}
  \vspace{0pt}

\begin{center}
\begin{tikzpicture}[scale=2]

\def\a{0.8}  
\def\b{1.6}

\def\p{20}
\def\t{80}

\draw[dashed] (0,0) ++(0:\a) arc[start angle=0, end angle=180, radius=\a];
\draw[dashed] (0,0) ++(0:\b) arc[start angle=0, end angle=180, radius=\b];
\draw (0,0) -- (\p:\b);
\draw (0,0) -- (\t:\b);

\path[fill=blue!20, draw=blue!60!black, thick]
	(\p:\a) -- (\p:\b) arc[start angle=\p, end angle=\t, radius=\b]
	-- (\t:\a) arc[start angle=\t, end angle=\p, radius=\a]
	-- cycle;

\node[right] at (\p:\b) {$r_1e^{i\phi_0}$};
\node[above right] at (\t:\b) {$r_1e^{i\phi_1}$};
\node[below right] at (\p:\a) {$r_0e^{i\phi_0}$};
\node[above left] at (\t:\a) {$r_0e^{i\phi_1}$};
\end{tikzpicture}
\end{center}
\caption{Annular sector $\ann(r_0,r_1;\phi_0,\phi_1)$.}
\label{fig:annular}
    
  \end{subfigure}
  \hfill
  \begin{subfigure}[t]{0.45\textwidth}
  \vspace{0pt}

\begin{center}
\begin{tikzpicture}[scale=2]

\def\a{0.8}  
\def\b{1.6}

\def\p{20}
\def\t{80}

\draw[dashed] (0,0) ++(0:\a) arc[start angle=0, end angle=180, radius=\a];
\draw[dashed] (0,0) ++(0:\b) arc[start angle=0, end angle=180, radius=\b];
\draw (0,0) -- (\p:\b);
\draw (0,0) -- (\t:\b);

\path[fill=blue!20, draw=blue!60!black, thick]
	(\p:\a) -- (\p:\b) arc[start angle=\p, end angle=\t, radius=\b]
	-- (\t:\a) arc[start angle=\t, end angle=\p, radius=\a]
	-- cycle;

\node[right] at (\p:{\b/cos((\t-\p)/2)}) {$A$};
\node[below right] at (\p:\a) {$D$};
\node[above left] at (\t:\a) {$C$};
\node[above right] at ((\t:{\b/cos((\t-\p)/2)}) {$B$};
\node[above left] at (0,0) {$O$};

\path[fill=red!20, fill opacity=0.3, draw=red!60!black, thick]
	(\p:\a) -- (\t:\a) -- (\t:{\b/cos((\t-\p)/2)}) 
	-- (\p:{\b/cos((\t-\p)/2)}) -- (\p:\a);

\end{tikzpicture}
\end{center}
\caption{Representation of the polygonal enclosure (in red) of an annular sector (in blue), following the same notation as in Lemma~\ref{lem:poly-enclosure}.}
\label{fig:poly-enclosure}
    
  \end{subfigure}
  \caption{Annular sectors, and their polygonal enclosures.}
  \label{fig:parallel}
\end{figure}
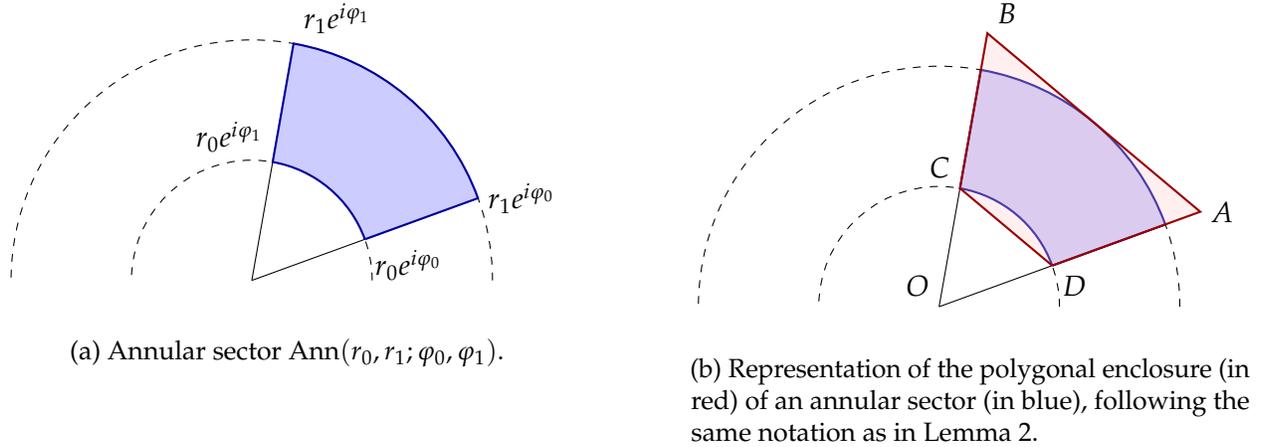

The lemma above provides a simple way of computing the Minkowski product of annular sectors. However, computing their Minkowski sum directly would not be straightforward. In order to be able to operate with these annular sectors, we will enclose them in convex polygons, motivated by the following result.

\begin{figure}
\end{figure}

\begin{lemma}
    Given some real numbers $0 \leq r_0 \leq r_1$ and $0 \leq \phi_0 \leq \phi_1 \leq 2\pi$ with $\phi_1 - \phi_0 < \pi$, consider the annular sector $\Gamma \coloneq \ann(r_0,r_1;\phi_0,\phi_1)$, and let $\Delta \coloneq \phi_1 - \phi_0$.
    
    Fixing the points
	\[ A \coloneq \frac{r_1}{\cos(\Delta/2)}e^{i\phi_0},\qquad
    B \coloneq \frac{r_1}{\cos(\Delta/2)}e^{i\phi_1},\qquad
	(C, D) \coloneq \left( r_0 e^{i\phi_1}, r_0 e^{i\phi_0}\right),\]

	we define the \emph{soft polygonal enclosure} of $\Gamma$ to be the polygon $P \coloneq \poly{A, B, C, D}$, which is represented in Figure~\ref{fig:poly-enclosure}. It satisfies the following properties:
	\begin{enumerate}[label=(\roman*)]
		\item The soft polygonal enclosure $P$ is a convex polygon.
		\item The annular sector $\Gamma$ is included in its soft polygonal enclosure $P$.
		\item The soft polygonal enclosure $P$ over-approximates the annular sector $\Gamma$ arbitrarily well as $\phi_1 \to \phi_0$. Formally, if $\mu(\cdot)$ denotes the surface area of a region in the complex plane,
        \[ \lim_{\phi_1 \to \phi_0} \frac{\mu(P)}{\mu(\Gamma)} = 1.\]

	\end{enumerate}

\label{lem:poly-enclosure}
\end{lemma}

\begin{proof} Let $O$ denote the origin ($O \coloneq 0 \in \mathbb{C}$).
\begin{enumerate}[label=(\roman*)]
\item It suffices to verify that, since $\phi_1 - \phi_0 \leq \pi$, the angles at all the vertices are smaller than or equal to $\pi$, and the result follows from Theorem~18.7 in~\cite{convex-sets}.

\item The intersection of the interior of $\Gamma$ with the area enclosed by $P$ can be easily proven to be nonempty, so it suffices to prove that the perimeter of $P$ lies in the exterior of $\Gamma$ or its boundary.

Define the points $A' = r_1 e^{i\phi_0}$ and $B' = r_1 e^{i\phi_1}$, in such a way that $AD = AA' \cup A'D$ and $BC = BB' \cup B'C$. Clearly, the segments $A'D$ and $B'C$ are in the boundary, whereas $AA'$ and $BB'$ lie in the exterior (excluding $A'$ and $B'$ themselves), as their moduli is greater than $r_1$.
Moreover, $CD$ (except for $C$ and $D$) lies in the exterior, as it is a chord joining two points in the circumference with radius $r_0$.

Lastly, regarding $AB$, it suffices to note that it is tangent to the circle of radius $r_1$, and hence lies in its exterior, with the exception of its midpoint, which, of course, is in the boundary.

\item On the one hand, the area of the annular sector is
\[ \mu(\Gamma) = \frac{\Delta}{2} (r_1^2 - r_0^2). \]
On the other hand, the area of the soft polygonal enclosure $P$ will be the area of the isosceles triangle $OAB$, minus the area of the isosceles triangle $ODC$. Moreover,
\[ \mu(OAB) = r_1^2 \tan(\Delta/2), \qquad \mu(ODC) = \frac{1}{2}r_0^2 \sin(\Delta),\]
therefore
\begin{align*}
\mu(P) = \mu(OAB) - \mu(ODC) &= r_1^2 \tan(\Delta/2)-\frac{r_0^2}{2}\sin(\Delta) \\
&= \frac{\Delta}{2}(r_1^2 - r_0^2) + \mathcal{O}(\Delta)^3,
\end{align*}
and the result follows.
\end{enumerate}
\end{proof}

The preceding result shows that the soft polygonal enclosure of an annular sector is a convex set that includes the entire sector and approximates it well as the sector becomes smaller. Moreover, since it is a polygon, it is much easier to operate with.
Nevertheless, we can only compute it for annular sectors in which the angle difference $\phi_1 - \phi_0$ is smaller than $\pi$, and the quality of the approximation improves as this angle difference decreases. Taking this into account, we can define a general polygonal enclosure as follows.

\begin{definition}
Let $S(r_0,r_1,\phi_0,\phi_1)$ denote the soft polygonal enclosure of an annular sector $\ann(r_0,r_1,\phi_0,\phi_1)$.

Consider a \emph{precision factor} $F \in \mathbb{N}$. Given any annular sector of the form $\Gamma \coloneq \ann(r_0,r_1,\phi_0,\phi_1)$ (with no bounds on the angles), let $\Delta = (\phi_1 - \phi_0)\bmod 2\pi$, and fix
\[N \coloneq \begin{cases}
F & \text{if }\Delta < \pi/8,\\
2F & \text{if }\Delta < \pi/4,\\
4F & \text{if }\Delta < \pi/2,\\
8F & \text{if }\Delta < \pi,\\
16F & \text{otherwise.}
\end{cases}\]
We define the polygonal enclosure of $\Gamma$ to be
\[ \polyenc{\Gamma} \coloneq \conv\left[{\bigcup_{r = 0}^{N-1} S\left(r_0, r_1, \phi_0 + (\phi_1 - \phi_0)\frac{r}{N}, \phi_0 + (\phi_1 - \phi_0) \frac{r + 1}{N}\right)}\right]. \]
\label{def:poly-enclosure}
\end{definition}

In the definition that we have just introduced, regardless of the choice of precision factor, we always have a well-defined polygonal enclosure, as the angles used in the soft polygonal enclosures are well below $\pi$. The precision factor can be adjusted at one's own discretion, in such a way that higher factor would lead to better approximations, at the cost of producing a more complex shape. We shall take $F \coloneq 4$.

Notice how, in the definition of polygonal enclosures, we have had to take the convex hull of the union of soft polygonal enclosures. Indeed, for annular sectors with non-zero measure, the union would never be convex, so the convex hull always yields a strict over-approximation.
However, for our purposes, it will provide a sufficiently accurate over-approximation (see Remark~\ref{note-convex-approximation}).
In order to efficiently compute convex hulls, one may use, for instance, the Monotone Chain algorithm~\cite{andrewanother}, which constructs the convex hull of a set of points sorting them by their real component (and, in a tie, by their imaginary one) and then computing the lower and upper hulls with a linear pass through the points. This is the algorithm that we will use in our implementation. Nevertheless, it is worth pointing out that another alternative is to use the Kirkpatrick--Seidel algorithm~\cite{purge-convex}, which is more intricate and thus not as well-suited for small input sizes like the ones we consider, but has better asymptotic complexity. In particular, if $n$ denotes the number of input points and $h$ represents the number of points in the convex hull, the Monotone Chain algorithm has an asymptotic complexity of $\mathcal{O}(n \log n)$, whereas Kirkpatrick--Seidel yields $\mathcal{O}(n \log h)$.

To conclude this section, we will discuss how to determine whether $0$ is contained in the Minkowski sum of some convex polygons of the complex plane. Given some convex polygons, it is always possible~\cite{efficient-sum} to efficiently compute the vertices of the convex polygon that corresponds to their Minkowski sum. Then, determining whether $0$ belongs to this sum is straightforward, as the following result shows.

\begin{lemma}
Consider a convex polygon $P \coloneq \poly{p_1,\ldots,p_\ell} = \conv\{p_1,\ldots,p_\ell\}$ in the complex plane. The point $0$ is contained in the interior or boundary of $P$ if and only if all the cross products $p_j \times (p_{j + 1} - p_j)$ have the same sign for $j = 1,\ldots, \ell$, setting $p_{\ell+1} \coloneq p_1$.
\end{lemma}

\begin{proof}
A convex polygon coincides with the intersection of the half-planes defined by its edges~\cite[\S 1--2]{polytopes}, hence, for a fixed $\sigma \in \{+1, -1\}$,
\[ P = \bigcap_{j=1}^\ell \left\{p\in\mathbb{C} \mid \operatorname{sign} \left((p_j - p) \times ( p_{j+1} - p_j) \right) \in \{\sigma, 0\}\right\}.\]
The result is immediate from this fact.
\end{proof}

\section{Theory solver}
\label{sec:tsolver}
Without loss of generality, consider a rank-one projector $\Pi \coloneq \ket{v}\bra{v}$ supported on $[k]$. Given some intervals
\begin{equation}
[\theta_{10},\theta_{11}], \ldots, [\theta_{k0},\theta_{k1}] \subseteq [0, \pi],\qquad
[\phi_{10}, \phi_{11}], \ldots, [\phi_{k0}, \phi_{k1}] \subseteq [0,2\pi],
\tag{$*$}
\label{intervals-tsolver}
\end{equation}
we will consider some one-qubit states $\ket{\psi_1}, \ldots, \ket{\psi_k}$ such that, for every $j = 1,\ldots,k$
\[\ket{\psi_j} = \ket{\psi(\theta_j, \phi_j)} = \cos(\theta_j/2)\ket{0} + e^{i\phi_j} \sin(\theta_j/2)\ket{1},\qquad
\theta_j \in [\theta_{j0}, \theta_{j1}], \qquad
\phi_j \in [\phi_{j0}, \phi_{j1}].\]

Our goal is to determine whether, under these hypotheses, it is ever possible for the product state $\ket{\psi} \coloneq \ket{\psi_1}\otimes\cdots \otimes \ket{\psi_k}$ to be orthogonal to $\ket{v}$ (i.e., in the kernel of $\ket{v}\bra{v}$). To this end, we will introduce an algorithm that will return two possible outputs:
\begin{itemize}
\item UN-PRODSAT if it can certify that $\bra{v} (\ket{\psi_1}\otimes\cdots\ket{\psi_k}) \neq 0$ for all allowed configurations,
\item and MAYBE($A$, $\rho$) otherwise, specifying the area $A$ and maximum modulus squared $\rho$ of a set that covers all the possible values of $\braket{v}{\psi}$, which includes $0$. 
\end{itemize}

In order to determine whether $\bra{v}(\ket{\psi_1} \otimes \cdots \otimes \ket{\psi_k}) = 0$ for all allowed configurations, we shall first consider, for every $j = 1,\ldots, k$, the interval
\[ P_j^0 \coloneq \cos([\theta_{j0}/2, \theta_{j1}/2]),\]
and the annular sector
\[P_j^1 \coloneq {\sin([\theta_{j0}/2, \theta_{j1}/2])\cdot \exp(i[\phi_{j0}, \phi_{j1}])},\]
It is immediate to verify that $\braket{0}{\psi_j} \in P_j^0$ and $\braket{1}{\psi_j} \in P_j^1$.

Using the resolution of the identity, we know that
\begin{align*}
\bra{v}(\ket{\psi_1} \otimes \cdots \otimes \ket{\psi_k}) &=\bra{v} I_{2^k} (\ket{\psi_1} \otimes \cdots \otimes \ket{\psi_k})\\
&= \bra{v} \left(\sum_{t = 0}^{2^k - 1} \ket{t}\bra{t} \right)(\ket{\psi_1} \otimes \cdots \otimes \ket{\psi_k}) \\
&= \sum_t \braket{v}{t} \cdot \braket{t_1}{\psi_1}\cdots\braket{t_k}{\psi_k}\\
&\in \bigoplus_t \left( \braket{v}{t} \cdot P_1^{t_1} \odot \cdots \odot P_k^{t_k} \right)
\end{align*}
where we use $t_1,\ldots, t_k$ to denote the $k$ binary digits of $t$, and where $\oplus$ and $\odot$ denote, respectively, the Minkowski sum and product of sets. Therefore, a sufficient condition for having $\braket{v}{\psi} \neq 0$ for every possible configuration would be
\begin{equation}
0 \not\in \bigoplus_t\left(  \braket{v}{t} \cdot P_1^{t_1} \odot \cdots \odot P_k^{t_k}  \right) .
\tag{S}
\label{suf-cond-0}
\end{equation}
This condition can be verified using the geometrical tools that we introduced back in Section~\ref{sec:geometry}. The Minkowski products can be computed using Lemma~\ref{lem:product-hull}, and their subsequent multiplication by the scalars $\braket{v}{t}$ is direct. This turns the right hand side of the expression into a Minkowski sum of annular sectors. Using Lemma~\ref{lem:poly-enclosure} and Definition~\ref{def:poly-enclosure}, we can find the polygonal enclosures of these sectors, and, as we discussed towards the end of Section~\ref{sec:geometry}, we can efficiently determine whether $0$ belongs to their Minkowski sum.

\begin{remark}
As we mentioned in Section~\ref{sec:geometry}, when defining the polygonal enclosure of an annular sector, we take the convex hull of a set that already covers the sector --- and the convex hull is always a strict superset. However, the convex hull still provides a good over-approximation if it is going to be part of a Minkowski sum with a large number of summands, by virtue of the Shapley--Folkman lemma~\cite{shapley-folkman}. This result establishes that, in a Minkowski sum of sets $S \coloneq A_1 \oplus \cdots \oplus A_\ell \subseteq \mathbb{C}$, every element $s$ in the convex hull $\conv{S}$ can be written as a sum $s = a_1 + \ldots + a_\ell$ where every element $a_j$ belongs to $A_j$ except for at most two indices $j_1,j_2$ for which $a_{j_1} \in \conv(A_{j_1}) \setminus A_{j_1}$ (and analogously for $a_{j_2}$).
\label{note-convex-approximation}
\end{remark}

This defines the behaviour of the theory solver: if, using the tools that we have introduced, the sufficient condition \eqref{suf-cond-0} for $\braket{v}{\psi} \neq 0$ is verified, the theory solver will return UN-PRODSAT. Otherwise, it will return MAYBE, together with the area and maximum modulus of the Minkowski sum in \eqref{suf-cond-0}, which includes $0$. Clearly, with an UN-PRODSAT result, we can have full certainty that under none of the allowed configurations can there be a state satisfying the constraint defined by $\ket{v}\bra{v}$. On the other hand, a MAYBE result does not guarantee the existence of a state satisfying the constraint among the allowed configurations --- although, in such a case, the certainty about the existence of a solution grows as the returned area becomes smaller.

\newcommand\css[1]{\operatorname{CSS}\left\langle#1\right\rangle}

The following result summarises and encapsulates the procedure that we have defined in this section.

\begin{proposition}
\label{prop:theory-solver}
Consider a constraint $\Pi$ with support $s_1,\ldots,s_k$ defined by a $k$-qubit state $\ket{v}$, and some intervals $[\theta_{j0}, \theta_{j1}] \subseteq [0,\pi]$ and $[\phi_{j0}, \phi_{j1}] \subseteq [0,2\pi]$ for $j = 1,\ldots k$, and define
\[ \Omega \coloneq \{\ket{\Psi(\theta_1,\phi_1)}\otimes\cdots\otimes\ket{\Psi(\theta_n, \phi_n)} \mid \theta_{s_j} \in [\theta_{j0}, \theta_{j1}], \phi_{s_j} \in [\phi_{j0}, \phi_{j1}], j = 1,\ldots,k\}\]
to be the set of $n$-qubit product states where the Bloch sphere angles of qubits $s_j$ are bounded by $[\theta_{0j}, \theta_{1j}]$ and $[\phi_{j0}, \phi_{j1}]$ for $j = 1,\ldots,k$.

Upon sending $v, (s_1,\ldots,s_k), ([\theta_{j_0}, \theta_{j_1}], [\phi_{j_0}, \phi_{j_1}])$ to Algorithm~\ref{alg:constraint}:
\begin{enumerate}[label=(\roman*)]
\item If the result is UN-PRODSAT, then no state in $\Omega$ satisfies the constraint $\Pi$.
\item If the result is MAYBE($A$, $\rho$), then the set of values $\Pi\ket{\psi}$ taken when $\ket{\psi} \in \Omega$ may contain zero and is covered by a set of area $A$ and maximum modulus squared $\rho$.
\end{enumerate}
\end{proposition}

\begin{algorithm}
\caption{Theory solver}
\begin{algorithmic}[1]
    \Input Constraint defined by a $k$-qubit state $\ket{v}$ with support $(s_1,\ldots,s_k)$.
    \Input Sequence of $k$ ordered pairs of intervals $([\theta_{j0}, \theta_{j1}], [\phi_{j0}, \phi_{j1}])_{j=1}^k$.
    \Statex

    \State polygon\_summands $\gets$ []
    \For{$t = (\overline{t_1,\ldots,t_k})_2\gets 0,\ldots2^k-1$}
         \Comment{$t_{(\cdot)}$ are the binary digits of $t$.}
        \State annular\_sectors $\gets$ []
        \For{$j\gets 1,\ldots,k$}
            \If{$t_j = 1$}
                \State append(annular\_sectors, ${\sin([\theta_{j0}/2, \theta_{j1}/2])\cdot \exp(i[\phi_{j0}, \phi_{j1}])}$)
            \Else
                \State append(annular\_sectors, $\cos([\theta_{j0}/2, \theta_{j1}/2])$)
            \EndIf
        \EndFor
        \State append(polygon\_summands, $\polyenc{\braket{v}{t} \cdot \bigodot \text{annular\_sectors}}$)
    \EndFor
    \State $P \gets \bigoplus \text{polygon\_summands}$
    \If{$0 \in P$}
        \State \Return MAYBE(Area($P$), $\max_{p\in P} \abs{p}^2$)
    \Else
        \State \Return{UN-PRODSAT}
    \EndIf
\end{algorithmic}
\label{alg:constraint}
\end{algorithm}

\section{Main architecture (Proof of Theorem~\ref{thm:main})}
\label{sec:main}

In this section, we construct Algorithm~\ref{alg:full} and prove Theorem~\ref{thm:main}.

\newcommand\var[1]{\hat{#1}}
\newcommand\bisect[1]{\operatorname{BSS}\left\langle#1\right\rangle}

Let $x_1,\ldots,x_\ell$ be $\ell$ boolean values, with $\ell \geq 1$. Setting
\[B(x_1,\ldots,x_\ell) \coloneq \sum_{j = 1}^\ell x_j\cdot\frac{2\pi}{2^j},\]
we define the interval
\[ I(x_1, \ldots, x_\ell) \coloneq B(x_1,\ldots,x_\ell) + \Ibox{0, \frac{2\pi}{2^\ell}} = \Ibox{B, \frac{2\pi}{2^\ell}+B}.\]

This interval effectively represents the results of performing $\ell$ iterations of a dichotomic search through $[0,2\pi]$, selecting the left half in the $j$th iteration if $x_j = 0$, and choosing the right half if $x_j = 1$.
Clearly, the set of all intervals of this form constitutes, up to an overlap in a measure-zero subset, a partition of $[0, 2\pi]$. Moreover, the collection of all intervals of the form $I(0, x_2,\ldots, x_\ell)$ (for a fixed $\ell$) partitions in an analogous manner the interval $[0,\pi]$.
In Figure~\ref{fig:intervals}, we show a visual representation of how these intervals are constructed.

\begin{figure}
\newcommand\interval[6]{
\draw[dashed] (-0.1,-#5) -- (1.1,-#5);
\draw[very thick,{Circle}-{Circle}, #6] (#1, -#5) node[above left] {$#3$} -- (#2, -#5) node[above right] {$#4$};
}

\centering
\begin{tikzpicture}[xscale=11, yscale=1.5]

\fill[BrickRed!15] (0,0) rectangle (0.5, -1);
\fill[MidnightBlue!15] (0.25,-1) rectangle (0.5, -2);
\fill[OliveGreen!15] (0.25, -2) rectangle (0.375, -3);
\node[BrickRed] at (0.25, -0.5) {$x_1 = 0$};
\node[MidnightBlue] at (0.375, -1.5) {$x_2 = 1$};
\node[OliveGreen, rounded corners, fill = OliveGreen!15] at (0.3125, -2.5) {$x_3 = 0$};

\interval{0}{1}{0}{2\pi}{0}{black}
\interval{0}{0.5}{0}{\pi}{1}{BrickRed}
\interval{0.25}{0.5}{\frac{\pi}{2}}{\pi}{2}{MidnightBlue}
\interval{0.25}{0.375}{\frac{\pi}{2}}{\frac{3\pi}{2}}{3}{OliveGreen}
\end{tikzpicture}
\caption{The black, red, blue and green segments are, respectively, the intervals $I()$, $I(0)$, $I(0,1)$ and $I(0,1,0)$. Each additional bit that is passed to $I$ represents a new step in a binary search through the interval $[0, 2\pi]$.}
\label{fig:intervals}
\end{figure}
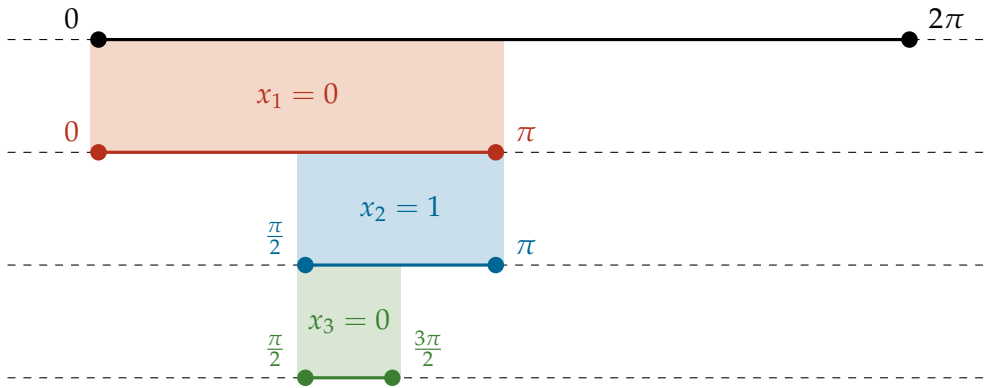

Consider a QSAT($k$) instance $Q$ as in Definition~\ref{def:pqsat}, and fix a \emph{search depth} $D \in \mathbb{N}$. We shall now formalise a procedure which will aim to determine whether there exist $2n$ intervals
\[\left\{\Theta_j \coloneq I(0,x_j^{1} \ldots, x_j^{D-1}), \quad \Phi_j \coloneq I(y_j^1, \ldots, y_j^D)\right\}_{j=1}^n,\]
with $x_{j}^{d}, y_j^d \in \mathbb{B}$, for which there may exist a solution $\ket{\psi}$ to the problem such that $\ket{\psi_j} = \ket{\Psi(\theta_j, \phi_j)}$ and $\theta_j \in \Theta_j, \phi_j \in \Phi_j$. The resulting algorithm will return UN-PRODSAT if it can certify that no such interval exists, and it will return MAYBE($A$, $\rho$) otherwise, as the theory solver. Moreover, under additional hypotheses on the problem instance, the algorithm may also return SAT if it can verify the existence of a satisfying assignment.

We model the problem as follows. We first consider the language of propositional logic over a set of $n \times (2D - 1)$ variables:
\[\mathcal{V} \coloneq \left\{\var{\theta}^{d'}_{j}, \var{\phi}^d_j \;\middle\vert\; j  \in [n],\; d \in[D],\; d' \in [D-1]\right\}.\]
For convenience, we will write
\[ \var{\theta}_j \coloneq (\var{\theta}_j^1,\ldots,\var{\theta}_j^{D-1}),\qquad
\var{\phi}_j \coloneq (\var{\phi}_j^1,\ldots,\var{\phi}_j^D).
\]

Following standard terminology, we define a \emph{literal} over our language to be a construction of the form $v$ or $\lnot v$ for $v \in \mathcal{V}$, and we take a \emph{clause} to be a disjunction of literals, and a \emph{(CNF) formula} to be a conjunction of clauses. Moreover, we define a \emph{total assignment} over our language to be a function $\sigma : \mathcal{V} \longrightarrow \mathbb{B}$, and we take a \emph{partial assignment} to be a function $\sigma' : D \subset \mathcal{V} \longrightarrow \mathbb{B}$. Given a partial or total assignment, the truth value (true, false or undefined) of clauses and formulas is defined in the obvious way.

We say that a total assignment $\sigma$ is \emph{$(Q,D)$-invalid} if, under the assumption that an $n$-qubit state is of the form
\[\ket{\psi} = \ket{\Psi(\theta_1, \phi_1)}\otimes\cdots\otimes\ket{\Psi(\theta_n,\phi_n)},\qquad \theta_j \in I(0,\sigma(\var{\theta}_j)),\quad \phi_j \in I(\sigma(\var{\phi}_j)),\]
there exists a constraint such that no product state in the induced region satisfies it; in this case, we also say that the constraint is not satisfied under the assignment $\sigma$. Otherwise, if all constraints are satisfied for a $\ket{\psi}$ under these conditions, we say that $\sigma$ is \emph{$(Q,D)$-valid}. Moreover, a clause is said to be \emph{$(Q,D)$-valid} if it is satisfied under any $(Q,D)$-valid total assignment.

Having introduced these basic definitions, the following result holds the key to our approach.
\begin{lemma}
Consider a QSAT instance $Q$, modelled as above.
\begin{enumerate}[label=(\roman*)]
\item If a (CNF) formula consisting of the conjunction of $(Q,D)$-valid clauses is unsatisfiable, then the QSAT instance $Q$ is UN-PRODSAT.\label{stm:unsatunprodsat}
\item If there exists a $(Q,D)$-valid total assignment, the hypothesis in \ref{stm:unsatunprodsat} can never hold.
\end{enumerate}
\label{lem:model}
\end{lemma}

\begin{proof}
    Regarding \ref{stm:unsatunprodsat}, if several $(Q,D)$-valid clauses cannot be satisfied simultaneously, that means that there cannot exist a $(Q,D)$-valid total assignment, hence $Q$ is UN-PRODSAT. Indeed, if $Q$ were PRODSAT, there would exist a satisfying $\ket{\Psi(\theta_1,\phi_1)} \otimes \cdots \ket{\Psi(\theta_n, \phi_n)}$, so a total assignment with\footnote{Such an assignment must exist, as the union of intervals of the form $I(x_1,\ldots,x_D)$ covers $[0,2\pi]$.} $\theta_j \in I(0,\sigma(\var{\theta}_j))$ and $\phi_j \in I(\sigma(\var{\phi}_j))$ would be $(Q,D)$-valid --- and therefore satisfy any $(Q,D)$-valid clause.

    The remaining statement is trivial.
\end{proof}

In addition to this, we will now discuss how every $(Q,D)$-invalid assignment induces a $(Q,D)$-valid clause that excludes that assignment.

For a total assignment $\sigma$ and a variable $v$, let $[v]^\sigma$ denote the literal $v$ if $\sigma(v) = 0$ and $\lnot v$ if $\sigma(v) = 1$. If a constraint with support $s_1, \ldots, s_k$ is not satisfied under $\sigma$, then, clearly
\[ \bigvee_{j = 1}^k \left(\bigvee_{d=1}^D [\var{\phi}_{s_j}^d]^\sigma\vee \bigvee_{d'=1}^{D-1} [\var{\theta}_{s_j}^d]^\sigma\right) \]
is $(Q,D)$-valid, as it is satisfied by any assignment that does not coincide with $\sigma$ over the variables that define the Bloch sphere angles on the qubits $s_1,\ldots,s_k$. Nevertheless, even if this clause is $(Q,D)$-valid, it is not particularly informative, as it includes $k(2D-1)$ variables. Luckily, we can potentially get rid of several of its literals.

If a certain constraint is not satisfied under a total assignment $\sigma$, that means that the constraint cannot be satisfied when the Bloch sphere angles of its qubits lie in
\[\theta_{s_j} \in I\left(0,\sigma(\var{\theta}_{s_j}^1), \ldots, \sigma(\var{\theta}_{s_j}^{D-1})\right), \qquad \phi_{s_j} \in I\left(\sigma(\var{\phi}_{s_j}^1), \ldots, \sigma(\var{\phi}_{s_j}^D)\right)\]
for $j\in[k]$. However, we can still run Algorithm~\ref{alg:constraint} repeatedly to find smaller values $D_1,\ldots,D_k$ and $D'_1,\ldots,D'_{k}$ such that Algorithm~\ref{alg:constraint} certifies UN-PRODSAT when the angles are restricted to
\[\theta_{s_j} \in I\left(0,\sigma(\var{\theta}_{s_j}^1), \ldots, \sigma(\var{\theta}_{s_j}^{D_j'})\right), \qquad \phi_{s_j} \in I\left(\sigma(\var{\phi}_{s_j}^1), \ldots, \sigma(\var{\phi}_{s_j}^{D_j})\right),\]
and hence the $(Q,D)$-valid blocking clause can be reduced to
\begin{equation}
\bigvee_{j = 1}^k \left(\bigvee_{d=1}^{D_j} [\var{\phi}_{s_j}^{d}]^\sigma \vee \bigvee_{d'=1}^{D'_j} [\var{\theta}_{s_j}^{d'}]^\sigma\right).
\tag{RC}
\label{eq:reduced-clause}
\end{equation}
In particular, we suggest using Algorithm~\ref{alg:generalise} for finding such values.
\begin{algorithm}[h]
\caption{Clause generalisation rule}
\begin{algorithmic}[1]
    \Input Constraint $\Pi$ supported on $s_1,\ldots,s_k$.
    \Input Total assignment $\sigma$, for which the theory solver returns UN-PRODSAT on $\Pi$.
    \Input Search depth $D$.
    \Statex

    \State We let $C(d_1,\ldots,d_k,d'_1,\ldots,d'_k)$ be the result of calling the theory solver on $\Pi$ with \[\theta_{s_j} \in I\left(0,\sigma(\var{\theta}_{s_j}^1), \ldots, \sigma(\var{\theta}_{s_j}^{d_j'})\right),\qquad\phi_{s_j} \in I\left(\sigma(\var{\phi}_{s_j}^1), \ldots, \sigma(\var{\phi}_{s_j}^{d_j})\right).\]
    \Statex

    \State Initialise some variables $(D_1,\ldots,D_k, D_1', \ldots,D_{k}') \gets (D, \ldots, D, D-1, \ldots, D-1)$.
    \State Initialise a vector of movable variables $M \gets (D_1,\ldots,D_{k}, D'_1, \ldots, D'_{k})$
    \State For any variable $d$, let $m(d) \gets d$ if $d\not\in M$ and $m(d) \gets d - 1$ otherwise.
    \Statex
    \While{$M\neq\emptyset$}
        \If{$C(m(D_1), \ldots, m(D_k), m(D_1'), \ldots, m(D_k'))$ is UN-PRODSAT}
            \State $(D_1,\ldots,D_k, D_1', \ldots, D_k') = (m(D_1), \ldots, m(D_k), m(D_1'), \ldots, m(D_k'))$
        \Else
            \For{$d$ in $M$}
            \Comment{We iterate over the variables in $M$, not their underlying values.}
                \State $d \gets d - 1$
                \If{$C(D_1, \ldots, D_k, D_1', \ldots, D_k')$ is MAYBE}
                    \State $d \gets d + 1$
                    \State Remove $d$ from $M$
                \EndIf
            \EndFor
        \EndIf
        \State Only keep in $M$ the variables $d$ with a strictly positive value.
    \EndWhile
    \State \Return $\bigvee_{j = 1}^k \left(\bigvee_{d=1}^{D_j} [\var{\phi}_{s_j}^{d}]^\sigma \vee \bigvee_{d'=1}^{D'_j} [\var{\theta}_{s_j}^{d'}]^\sigma\right)$
\end{algorithmic}
\label{alg:generalise}
\end{algorithm}

Taking all of this into account, there is a natural way of tackling our problem:
\begin{enumerate}
    \item Consider an initial total assignment $\sigma$ and initialise a database of $(Q,D)$-valid blocking clauses.
    \item \label{step:base} Check $\sigma$ on every constraint using Algorithm~\ref{alg:constraint}.
    \begin{enumerate}
        \item If all the constraints output MAYBE($A_r, \rho_r$), by Lemma~\ref{lem:model}, return MAYBE($\sum_r A_r, \sum_r \rho_r$).
        \item If some constraints are not satisfied by the current total assignment, add to the database of blocking clauses all the $(Q,D)$-valid clauses produced by the clause generalisation rule (Algorithm~\ref{alg:generalise}) for each of those constraints. All of these clauses will, in particular, exclude this assignment.
    \end{enumerate}
    \item \label{step:sat} Run a CDCL solver to find a new total assignment that will satisfy all the $(Q,D)$-valid blocking clauses.
    \begin{enumerate}
        \item If the CDCL SAT solver returns a (potentially partial) assignment, fill it into a complete assignment $\sigma$ and go back to Step~\ref{step:base}.
        \item If the CDCL SAT solver concludes that the set of clauses it was given is unsatisfiable, return UN-PRODSAT.
    \end{enumerate}
\end{enumerate}

In the algorithm above, we only return UN-PRODSAT when we can find an unsatisfiable set of $(Q,D)$-valid clauses, as in Lemma~\ref{lem:model}\ref{stm:unsatunprodsat}. Conversely, we return MAYBE whenever may have found a $(Q,D)$-valid total assignment, for in this case we cannot certify UN-PRODSAT. Moreover, we return two values, $A$ and $\rho$, which corresponds to the addition of the surface areas and maximum moduli squared of the sets produced by the theory solver. Clearly, the likelihood of having a SAT assignment grows as these values become smaller.

It is worth highlighting that our algorithm always comes to a halt, as, for each total assignment $\sigma$ found in Step~\ref{step:base}, either the algorithm returns or it incorporates a new clause that prevents that assignment from ever being chosen again in Step~\ref{step:sat}.

In terms of worst-case complexity, our algorithm runs, for a fixed $k$, in a time that is exponential on the number of qubits. Indeed, every call to the theory solver and the production of a blocking clause runs in time $K$ only dependent on $k$ and not on the number of qubits. Moreover, we only have $2^{n(2D-1)}$ possible assignments to test and, thus, that is the maximum number of iterations that the algorithm may go through. Together with the fact that the worst-case complexity of the SAT solver is also $2^{n(2D-1)}$, we are left with a total worst-case complexity of $K\cdot 4^{n(2D-1)}$.

This completes the proof of Theorem~\ref{thm:main}. Algorithm~\ref{alg:full} summarises the procedure that we have introduced.

\algnewcommand\algorithmiclabel{\textbf{label}}
\algnewcommand\Label[1]{\State \algorithmiclabel\ #1}

\algnewcommand\algorithmicgoto{\textbf{goto}}
\algnewcommand\Goto[1]{\State \algorithmicgoto\ #1}
\begin{algorithm}
\caption{PRODSAT-QSAT($k$)}
\begin{algorithmic}[1]
    \Input Instance $Q$, with a set of $k$-local constraints over $n$ qubits.
    \Input Search depth $D$.
    \Statex

    \State CDCL-SAT $\gets$ Initialise a CDCL-SAT solver instance over the variables in $\mathcal{V}$.
    \State $\sigma \gets$ Initialise a total assignment on all the binary variables in $\mathcal{V}$.
    \Statex
    \Label{START}
    \If{theory solver outputs MAYBE($A_j, \rho_j$) on every constraint under $\sigma$}
        \State\Return MAYBE($\sum_j A_j$, $\sum_j \rho_j$)
    \Else
        \For{every constraint $\Pi$ on which the theory solver outputs UN-PRODSAT}
            \State $C \gets$ application of the clause generalisation rule on $\Pi$ and $\sigma$ (with search depth $D$)
            \State Pass the clause $C$ to CDCL-SAT 
        \EndFor
    \EndIf
    \Statex
    
    \If{CDCL-SAT returns UNSAT}
        \State \Return UN-PRODSAT
    \Else
        \State $\sigma\gets$ total assignment satisfying all the blocking clauses
        \Goto{START}
    \EndIf
\end{algorithmic}
\label{alg:full}
\end{algorithm}

Putting everything together, we have defined a procedure that will return, for any QSAT($k$) instance, an UN-PRODSAT result (in which case the instance is, indeed, UN-PRODSAT) or an inconclusive MAYBE($A$, $\rho$) result --- where, the smaller the $A$ and $\rho$ values, the more likely it is that the instance will be PRODSAT. In this second scenario, one can resort to the Buchberger algorithm~\cite{prodsat2024}, which has a doubly-exponential worst-case computational cost, to get a definite result.

\subsection{Additional hypotheses for PRODSAT results}
With additional hypotheses, we may be able to certify PRODSAT with our own method.

\begin{proposition}
Consider a QSAT($k$) instance with some constraints $\Pi_j$ for $j \in [m]$, and define the Hamiltonian
\[ H \coloneq \sum_{j=1}^m \Pi_j.\]
Under the hypothesis that either the instance is PRODSAT (resp. SAT) or the smallest eigenvalue of $H$ is bigger than a certain $E_0$, we can certify PRODSAT (resp. SAT) if our procedure outputs MAYBE($-,\rho$) and $\rho < E_0$.
\end{proposition}

\begin{proof}
If the algorithm outputs MAYBE($-,\rho$), then we know that there exists a region of the product state space in which, for every state $\ket{\psi}$,
\[ \bra{\psi}H\ket{\psi}  = \sum_j \bra{\psi}{\Pi_j}\ket{\psi} \leq \sum_j \rho_j = \rho,\]
hence if $\rho < E_0$, then $E_0$ cannot be the smallest eigenvalue of $H$.
\end{proof}

\subsection{Regarding the need for a SAT solver}
In our algorithm, we use a SAT solver to look for a total assignment satisfying a set of blocking clauses. The following result shows why our choice to use a SAT solver is well-justified.

\begin{proposition}
Let $k \leq n$ be a pair of natural numbers.
Assume that there exists a procedure that, given any set of clauses of the form \eqref{eq:reduced-clause} for a Quantum $k$-SAT instance on $n$ qubits, can return a boolean assignment satisfying those clauses or conclude that they are unsatisfiable. This procedure can solve Boolean $k$-SAT on $n$ variables.
\end{proposition}

\begin{proof}
To simplify our notation, given a boolean variable $X$, we shall represent its two literals as $(X,0) \coloneq \lnot X$ and $(X,1) \coloneq X$.

To prove our result, it suffices to show that any set of Boolean $k$-clauses on $n$ variables can be encoded as clauses of the form \eqref{eq:reduced-clause} when considering a Quantum $k$-SAT instance on $n$ qubits. Thus, let  $\{x_i\}_{i\in [n]}$ be Boolean variables, and fix an arbitrary set of $m$ $k$-clauses
\[ C\coloneq \left\{(x_{\iota(t, 1)}, \alpha(t,1)) \lor \cdots \lor (x_{\iota(t,k)}, \alpha(t, k))\right\}_{t=1}^m,\]
where $\iota : [m] \times [k] \longrightarrow [n]$ and $\alpha : [m]\times[k] \longrightarrow \{0,1\}$. Without loss of generality, we will assume that no variable has multiple occurrences in the same clause.

Let us now consider a Quantum $k$-SAT instance on $n$ qubits with $m$ constraints $\{\Pi_t\}_{t \in [m]}$, in such a way that, for each $t \in [m]$, the support of $\Pi_t$ is $\iota(t, 1), \ldots, \iota(t, k)$. Under such an instance, with any search depth $D$, any clause in $C$ can be of the form \eqref{eq:reduced-clause} up to a change of variables.
Indeed, for any constraint $\Pi_t$ supported on the qubits $i_1, \ldots, i_k$, any clause of the form
\[ \bigvee_{j=1}^k \left( \bigvee_{d = 1}^{D_j} (\var{\phi}_{i_j}^d, a_{jd}) \vee \bigvee_{d=1}^{D_j'} (\var{\theta}_{i_j}^d, b_{jd})\right)\]
can be a blocking clause for any $D_j, D_j' \in \mathbb{N}_0$ and any $a_{jd}, b_{jd} \in\{0,1\}$. This is true, in particular, for $D_j = 1$ and $D_j' = 0$, hence any clause of the form
\[ (\var{\phi}_{i_1}^1, a_1) \lor \cdots \lor (\var{\phi}_{i_k}^1, a_k)\]
can be a blocking clause. Particularising this for each of the constraints $\Pi_t$ and picking the values $a_r$ appropriately, it follows that 
\[ \left\{ \left(\var{\phi}_{\iota(t, 1)}^1, \alpha(t, 1)\right) \lor \cdots \lor \left(\var{\phi}_{\iota(t, k)}^1, \alpha(t, k)\right) \right\}_{t = 1}^m \]
is a set of clauses of the form \eqref{eq:reduced-clause} for this instance of Quantum $k$-SAT. Replacing the variables $x_i$ by $\var{\phi}_i^1$, the result follows.
\end{proof}

\section{Implementation}
\label{sec:implement}
We have implemented our algorithm in Rust, and our code is available on the following repository:
\begin{center}
    {\url{https://github.com/System-Verification-Lab/cdcl-qsat}}
\end{center}
It can run the algorithm numerically using floating-point numbers, or using an exact representation of elements in cyclotomic fields.
In fact, it can use any number representation provided by the user through the implementation of a trait.
Furthermore, it is possible to run the algorithm using any SAT solver. Through the RustSAT library~\cite{rustsat}, we provide built-in support for Glucose~\cite{glucose} (which is based on Minisat~\cite{minisat}), CaDiCaL~\cite{cadical}, and BatSat~\cite{batsat} (which is a Rust reimplementation of Minisat). Of course, any user can bring their own SAT solver by implementing the corresponding traits.

We have tested this implementation on randomly-generated states and verified that it is able to produce UN-PRODSAT results. In particular, using a search depth of $D = 8$, we tested it on $n$ qubits and $m$ constraints for $n =3,\ldots,9$ and $m = 1,\ldots,n+1$, generating $13$ random sets of constraints for every configuration of $n$ and $m$. The evolution of the runtime in terms of $n$ (for the case $n = m$), can be found in Figure~\ref{fig:time}.

\begin{figure}
\centering
\includegraphics[width=0.75\textwidth]{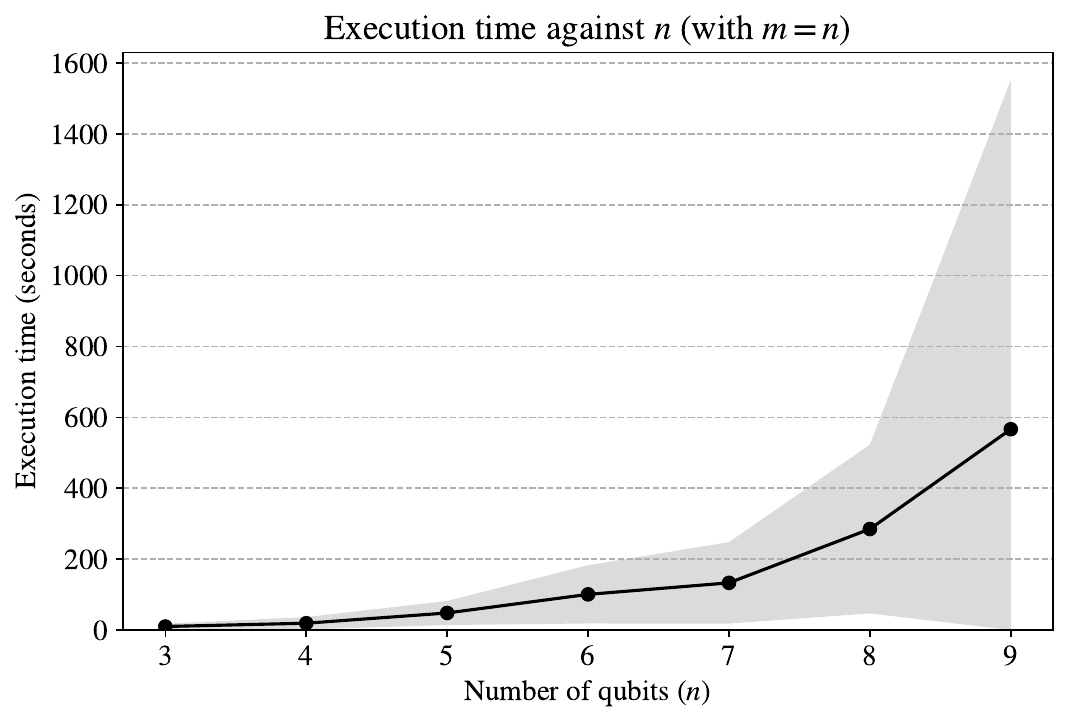}
\caption{Evolution of the runtime of our implementation (average and standard deviation), with respect to the number of qubits. To give a fair representation, we only consider instances in which the number of constraints $m$ is equal to the number of qubits $n$.}
\label{fig:time}
\end{figure}

\begin{table}
\centering
\begin{tblr}{
  colspec = {Q[c, gray!25]Q[c, gray!15] c c c c c c c c c c},
  hlines, vlines,
  cell{1}{1} = {r=2, c = 2}{m},
  cell{3}{1} = {r=7, c=1}{m},
  cell{1}{3} = {r=1, c=10}{m},
  row{1} = {bg = gray!25},
  row{2} = {bg = gray!15}
}
& & Number of constraints $m$ \\
&& 1 & 2 & 3 & 4 & 5 & 6 & 7 & 8 & 9 & 10\\
$n$ & 3 & 0/13 & 0/13 & 0/13 & 13/13 &  &  &  &  &  &  \\
& 4 & 0/13 & 0/13 & 0/13 & 1/13 & 13/13 &  &  &  &  &  \\
& 5 & 0/13 & 0/13 & 0/13 & 0/13 & 0/13 & 13/13 &  &  &  &  \\
& 6 & 0/13 & 0/13 & 0/13 & 0/13 & 0/13 & 1/13 & 12/13 &  &  &  \\
& 7 & 0/13 & 0/13 & 0/13 & 0/13 & 0/13 & 0/13 & 7/13 & 13/13 &  &  \\
& 8 & 0/13 & 0/13 & 0/13 & 0/13 & 0/13 & 0/13 & 1/13 & 4/13 & 13/13 &  \\
& 9 & 0/13 & 0/13 & 0/13 & 0/13 & 0/13 & 0/13 & 0/13 & 0/13 & 4/13 & 13/13 \\
\end{tblr}

\caption{Proportion of UN-PRODSAT certificates that the implementation of our algorithm was able to produce. These are consistent with theoretical expectations.}
\label{tab:unsatbench}
\end{table}
\begin{table}
\centering
\begin{tblr}{
  colspec = {Q[c, gray!25]cccccQ[c, gray!15]},
  hlines, vlines,
  row{1} = {bg = gray!25}
}
$n$ & {Avg.\\ \#clauses} & {Max.\\ \#clauses} & {Avg.\\ \#calls} & {Max. \\\#calls} & {Avg. \#learnt\\ clauses (CDCL)} & {Search\\space} \\
3 & \num{12140} & \num{88284} & \num{270951} & \num{1548287} & \num{1415} & $2^{45}$ \\
4 & \num{20429} & \num{199770} & \num{466371} & \num{3643611} & \num{2939} & $2^{60}$ \\
5 & \num{41796} & \num{316550} & \num{910593} & \num{5749351} & \num{6152} & $2^{75}$ \\
6 & \num{85866} & \num{753176} & \num{1824031} & \num{13917850} & \num{11078} & $2^{90}$ \\
7 & \num{94706} & \num{460644} & \num{2152597} & \num{17073270} & \num{11812} & $2^{105}$ \\
8 & \num{130422} & \num{968536} & \num{2761632} & \num{23582225} & \num{15211} & $2^{120}$ \\
9 & \num{189815} & \num{3016533} & \num{4010449} & \num{54923198} & \num{20783} & $2^{135}$ \\
\end{tblr}
\caption{For every number of qubits $n$, the average number of blocking clauses, calls to the theory solver and learnt conflict clauses in the CDCL solver (at the end of the execution) that were needed in every instance. We also show the maximum number of clauses and constraint calls that were needed to evaluate an instance.}
\label{tab:metrics}
\end{table}

In Table~\ref{tab:unsatbench}, we record the proportion of UN-PRODSAT certificates that our algorithm was able to produce in our benchmarks. These results align well with theoretical expectations~\cite{prodsat2024}. In Table~\ref{tab:metrics}, one can find metrics related to the number of blocking clauses and the number of calls to the theory solver, in addition to the number of learnt clauses kept by the CDCL solver at the end of the execution. Taking everything into account, these metrics show that the computational cost of our method is, on average, far better than its theoretical worst case. They also demonstrate the effectiveness of our theory solver and of our use of SAT solvers.

\section{Conclusions}
\label{sec:concl}

In this work we have introduced a search-driven clause-learning framework for PRODSAT-QSAT($k$), aimed at certifying the absence of satisfying product states for instances of quantum $k$-SAT with rank-one local projectors. The central idea is to discretise each single-qubit Bloch manifold into finitely many angular regions via a dichotomic partition, and to let a Boolean CDCL core search over these regions while a continuous theory solver checks feasibility of each constraint within the corresponding product region. The theory solver operates directly in the complex plane: by enclosing annular sectors with convex polygons and combining them through Minkowski sums and products, we obtain a computable over-approximation of the set of attainable local amplitudes, from which exclusion of the origin yields a sound UN-PRODSAT conclusion for that region.

The interaction between both components is mediated by $(Q,D)$-valid blocking clauses. When a total assignment is found to be $(Q,D)$-invalid (in the sense that some constraint cannot be satisfied anywhere inside the induced angular region), we produce a clause that excludes that region and add it to the clause database. Lemma~\ref{lem:model} encapsulates the resulting soundness argument: an unsatisfiable conjunction of $(Q,D)$-valid clauses certifies that the original instance is UN-PRODSAT. At the same time, when no constraint can be refuted on a candidate region, our procedure returns MAYBE($A, \rho$), where $A$ and $\rho$ quantify, respectively, the area and maximum modulus squared of the relevant set over-approximations. These quantities yield a heuristic indication of how ``tight'' the remaining uncertainty is: smaller values correspond to stronger evidence that the instance is PRODSAT, while any reported UN-PRODSAT result is a genuine result.

We have also provided a practical algorithm (Algorithm~\ref{alg:full}) together with an implementation in Rust supporting both floating-point arithmetic and exact representations over cyclotomic fields. Our experiments indicate that the method can produce UN-PRODSAT results in qualitative agreement with expected product-state satisfiability regimes. Nevertheless, the worst-case complexity remains exponential in the number of qubits at fixed depth, and our approach inevitably leads to inconclusive outcomes on some instances. In such cases, one may still resort to exact algebraic methods (e.g. Buchberger’s algorithm~\cite{prodsat2024}) to obtain a definitive answer, albeit at a significantly higher worst-case cost.

Several avenues may improve both performance and conclusiveness. On the theory side, tighter enclosures could reduce the prevalence of MAYBE outcomes, particularly by exploiting correlations between the two amplitudes of a single-qubit state that are currently relaxed by independent interval bounds. Extending these ideas beyond the product-state ansatz towards more expressive structured families, remains an appealing direction for future work.

\section*{Acknowledgements}
Joon Hyung Lee acknowledges support from the Dutch Research Council (NWO) as part of the project \emph{Boosting the Search for New Quantum Algorithms with AI (BoostQA)}, file number NGF.1623.23.033, under the research programme \emph{Quantum Technologie 2023}.

Samuel Gonz\'alez-Castillo has been supported by an FPU grant (Training Programme for Academic Staff, FPU22/02020) from the Spanish Ministry of Universities. He was also financially supported by the Ministry of Economic Affairs and Digital Transformation of the Spanish Government through the Spanish National Institute of Cybersecurity (INCIBE) project call for Strategic Projects on Cybersecurity in Spain, and by the European Union through the Recovery, Transformation and Resilience Plan – NextGenerationEU within the framework of the Digital Spain 2026 Agenda.

\begin{center}
\includegraphics[width=0.6\textwidth]{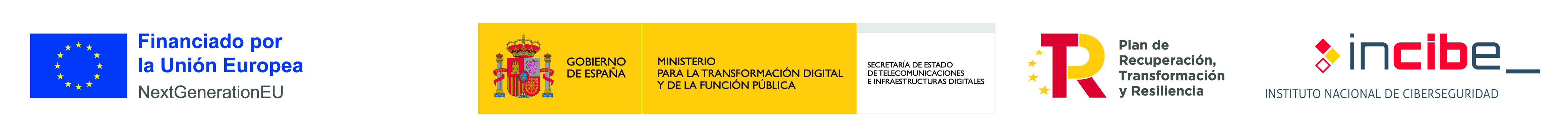}
\end{center}

\bibliographystyle{plain}
\bibliography{refs}

\end{document}